\documentclass[11pt]{amsart}

\usepackage[usenames]{color}

\usepackage{amsmath}
\usepackage{amsthm}
\usepackage{amsfonts}
\usepackage{amssymb}
\usepackage{latexsym}
\usepackage{color}
\usepackage{graphicx}
\usepackage{float}

\usepackage{blindtext}

\usepackage{hyperref}
\usepackage[top=1in, bottom=1.5in, left=1in, right=1in]{geometry}

\newtheorem{definition}{Definition}
\newtheorem{proposition}{Proposition}
\newtheorem{remark}{Remark}

\newcommand{\Z}{\mathbb{Z}}

\newcommand{\F}{\mathbb{F}}

\begin{document}

\author[]{Nael Rahman}
\address{Department of Mathematics, The City College of New York, New York,
NY 10031} \email{naelrahman@gmail.com}

\author[]{Vladimir Shpilrain}
\address{Department of Mathematics, The City  College  of New York, New York,
NY 10031} \email{vshpilrain@ccny.cuny.edu}

\title{MAKE: \lowercase{a \uppercase{M}atrix \uppercase{A}ction \uppercase{K}ey \uppercase{E}xchange}}

%

\begin{abstract}
We offer a public key exchange protocol based on a semidirect product of two cyclic (semi)groups of matrices over $\Z_p$.
One of the (semi)groups is additive, the other one multiplicative. This allows us to take advantage of both operations on matrices to diffuse information. We note that in our protocol, no power of any matrix or of any element of $\Z_p$ is ever exposed, so all standard attacks on Diffie-Hellman-like protocols  are not applicable.

\end{abstract}

\maketitle

\section{Introduction}

We start by recalling the classical Diffie-Hellman protocol \cite{DH}.
The simplest, and original, implementation of this protocol uses the
multiplicative group of integers modulo $p$, where $p$ is prime and
$g$ is primitive modulo $p$. A more general description of the
protocol uses an arbitrary finite cyclic group:

\begin{enumerate}

\item Alice and Bob agree on a finite cyclic group $G$ of order $q$ and a generating element $g$ in $G$.
 We will write the group $G$ multiplicatively.

\item Alice picks a random natural number $m<q$ and sends $g^m$ to Bob.

\item    Bob picks a random natural number $n<q$ and sends $g^n$ to Alice.

\item   Alice computes $K_A=(g^n)^m=g^{nm}$.

\item  Bob computes $K_B=(g^m)^n=g^{mn}$.
\end{enumerate}

Since $mn=nm$, both Alice and Bob are now in possession of the same
group element $K=K_A= K_B$ which can serve as the shared secret key.

The protocol is considered secure against eavesdroppers if $G$ and
$g$ are chosen properly. The eavesdropper must solve the {\it
Diffie-Hellman problem} (recover $g^{mn}$ from $g$, $g^m$, and $g^n$)
to obtain the shared secret key. This is currently considered
difficult for a ``good" choice of parameters (see e.g.
\cite{Menezes} for details).

In \cite{Habeeb},  a new  key exchange protocol was offered, based on a semidirect product of multiplicative matrix semigroups. That protocol is similar to the Diffie-Hellman protocol, but it differs in one essential
detail: at the last two steps, Alice and Bob use multiplication instead of exponentiation. In other words, the equality $K_A= K_B$ in that case is based on the identity $m+n=n+m$ in the ring of integers, instead of the
identity $mn=nm$. Also,  even though the parties do compute a large power of a public element (as in the classical Diffie-Hellman protocol), they do not transmit the whole result, but
rather just part of it, so no power of a public element is ever exposed.

The generic protocol in \cite{Habeeb} can be based on any (semi)group, in particular on any non-commutative (semi)group. For a particular instantiations of the generic protocol, the authors of \cite{Habeeb}
used conjugation as an action of the multiplicative semigroup of matrices on itself, and matrices were considered over the group ring $\Z_7[A_5]$. A ring of such a small size was selected to enhance efficiency; however, because of this small size,  the protocol turned out to be vulnerable to a ``linear algebra attack", similar to the
attack on Stickel's protocol \cite{Stickel_self} offered in \cite{Stickel}, albeit more sophisticated, see \cite{MR}, \cite{R}.

In this paper, we use matrices over $\Z_p$ for a large $p$. The size of the exponents used in the protocol should guarantee security against brute force attacks. We also note that in our protocol, no power of any matrix is ever exposed, so all standard attacks on Diffie-Hellman-like protocols (including Shor's quantum algorithm attack \cite{Shor}) are not applicable.

\section{Semidirect products and extensions  by automorphisms}
\label{Semidirect}

We include this section to make the exposition more comprehensive.
The reader who is uncomfortable with group-theoretic constructions
can skip to subsection \ref{holomorph}.

We now recall the definition of a semidirect product:

\begin{definition} Let $G, S$ be two groups, let $Aut(G)$ be the group of automorphisms of $G$,
and let $\rho: S \rightarrow Aut(G)$ be a homomorphism. Then the
semidirect product of $G$ and $S$ is the set
$$\Gamma = G \rtimes_{\rho} S = \left \{ (g, h): g \in G, ~h \in S \right \}$$
with the group operation given by\\
\centerline{$(g, h)(g', h')=(g^{\rho(h')} \cdot  g', ~h \cdot h')$.}\\
Here $g^{\rho(h')}$ denotes the image of  $g$ under the automorphism
$\rho(h')$, and when we write a product $h \cdot h'$ of two
morphisms, this means that $h$ is applied first.
\end{definition}

In this paper, we focus on a special case of this construction,
where the group $S$ is just a sub(semi)group of the semigroup of self-homomorphisms of $G$, typically a cyclic sub(semi)group. We give more details about
this special case in Section \ref{holomorph} below.

\subsection{Extensions  by automorphisms} \label{holomorph}

A particularly simple special case of the semidirect product
construction is where the group $S$ is just a subgroup of the group
$Aut(G)$ of automorphisms of $G$. Then the semidirect product is the set of all pairs $(g,
~\phi)$, where $g \in G, ~\phi \in S$, with the group operation
given by

\begin{align} \label{phi}
(g,~\phi)\cdot  (g',~\phi') = (\phi'(g)\cdot g',~\phi
\cdot \phi').
\end{align}

One can also use this construction if $G$ is not necessarily a
group, but just a semigroup, and/or consider endomorphisms (i.e.,
self-homomorphisms) of $G$, not necessarily automorphisms. Then the
resulting  semidirect product will be a semigroup, not a group, but this is quite sufficient for being the  platform of a Diffie-Hellman-like key exchange protocol.

\section{Action} \label{Action}

Our platform semigroup will be a semidirect product of two semigroups of matrices over $\Z_p$. In the notation of Section \ref{holomorph}, the semigroup $G$ will be  the {\it additive} semigroup of all matrices $\Z_p$, and the semigroup $S$ will be a cyclic {\it multiplicative}  semigroup of matrices. More specifically, the semigroup $S$ will consists of {\it pairs} of matrices over $\Z_p$, of the form $(H_1^k, H_2^k)$ for some fixed matrices $H_1, H_2$ and all positive exponents $k$.

The action of $S$ on $G$ will be as follows: $M^(S_1, S_2)=S_1MS_2$ for $M \in G, ~(S_1, S_2) \in S$. Note that this is an action only if we restrict to a {\it commutative} (in particular, cyclic) (semi)group  $S$ generated by a pair of matrices $(H_1, H_2)$.

Thus, in the semidirect product of the additive semigroup of matrices with the multiplicative cyclic semigroup generated by $(H_1, H_2)$, the multiplication looks like this (cf. formula (\ref{phi})):
\smallskip

\noindent $(M, (H_1, H_2))^2 = (M, (H_1, H_2)) \cdot (M, (H_1, H_2)) = (H_1MH_2+M, (H_1^2, H_2^2))$;\\
$(M, (H_1, H_2))^3 = (H_1MH_2+M, (H_1^2, H_2^2)) \cdot (M, (H_1, H_2)) = (H_1^2MH_2^2+H_1MH_2+M, (H_1^3, H_2^3))$, etc.

\section{Protocol description}\label{Protocol}

Below is the protocol description. Parameters are discussed separately, in Section \ref{Parameters}.
\smallskip

\begin{itemize}

\item[1.] ({\it key selection}) {\bf (i)} Alice and Bob agree on three public matrices, $M$, $H_1$,
and $H_2$, over $\Z_p$, such that $MH_i \ne H_iM$ and such that $det(H_1) = det(H_2)=0$.

\smallskip

\noindent {\bf (ii)} Alice selects a private integer $m$ and Bob selects a private integer $n$.

\smallskip

\item[2.]   Alice computes $(M, (H_1, H_2))^m$ and sends {\bf only the first component} (call it $A$) of the
result to Bob.

\smallskip

\item[3.]   Bob computes $(M, (H_1, H_2))^n$ and sends {\bf only the first component} (call it $B$) of the
result to Alice.

\smallskip

\item[4.]  Alice computes $(B, x) \cdot (A, ~(H_1, H_2)^m) = (H_1^m B H_2^m +A, ?)$. Her key is now $K_A = H_1^m B H_2^m +A$.
\smallskip

\item[5.]   Bob computes $(A, y) \cdot (B, ~(H_1, H_2)^n) = (H_1^n A H_2^n +B, ?)$.  His key is now $K_B = H_1^n A H_2^n +B$.

\smallskip

\item[6.]  Since $(M, (H_1, H_2))^{m+n} = (B, x) \cdot (A, ~(H_1, H_2)^m) = (A, ~y) \cdot (B, ~(H_1, H_2)^n) =
(K, ~(H_1, H_2)^{m+n})$, we should have $K_A = K_B = K$, the shared secret
key.

\end{itemize}

\begin{remark}
Note that, in contrast with the ``standard" Diffie-Hellman key
exchange, correctness here is based on the equality $x^{m}\cdot
x^{n} = x^{n} \cdot x^{m} =  x^{m+n}$  rather  than on the equality
$(x^{m})^{n} = (x^{n})^{m} = x^{mn}$. In  the ``standard"
Diffie-Hellman set up, our trick would not work because, if the
shared key $K$ was just the product of two  openly transmitted
elements, then anybody, including the eavesdropper, could compute
$K$.
\end{remark}

\section{Parameters and keys sampling}
\label{Parameters}

The basic parameter $p$ should have the same properties  and the same magnitude as recommended for the classical Diffie-Hellman protocol. In particular, $p$ should be a {\it safe prime}, i.e., a prime of the form  $p=2q+1$.

The exponents $m$ and $n$ should be of the same magnitude as $q=\frac{p-1}{2}$.

For the public matrices $M$, $H_1$,
and $H_2$, we require that $MH_i \ne H_iM$ and $det(H_1) = det(H_2)=0$. Sampling of non-commuting matrices can be done as follows. First build a matrix $M$ with entries selected uniformly at random from $\Z_p$. 

Then, the matrix $H_1$ is built as follows. First, make a diagonal matrix $D$: put 0 in the upper left corner, and random nonzero elements of $\Z_p$ in the remaining places on the diagonal. Then make sure that these random elements do not have order 2. If they do not, then, since $p-1=2q$, by Lagrange's theorem their order is at least $q$,  i.e., is large. Now select a matrix $S$ all of whose entries are random elements of $\Z_p$. With high probability, this matrix will be invertible. If not, then change one of the entries of $S$. If $S$ is invertible, let $H_1=S^{-1}DS$. 

To build $H_2$, use the same procedure but with fresh randomness throughout.

Finally, check that $MH_i \ne H_iM$; this will be the case with high probability. If this is not the case for one of the $H_i$, redo selection of $M$.


\section{Security}
\label{Security}

In this section, we address security of the protocol described in Section \ref{Protocol}.

Our security assumption here, analogous to the {\it computational Diffie-Hellman assumption}, is that it is computationally infeasible
to retrieve the shared secret key $K$ from the 5 public matrices $(M, H_1, H_2, A, B)$. The matrices $A$ and $B$ are expressed in terms of the matrices $M$, $H_1$, and $H$ as follows:

\begin{align} \label{A}
A =  H_1^{m-1}MH_2^{m-1} + H_1^{m-2}MH_2^{m-2}  + \ldots + H_1MH_2 + M.
\end{align}

\begin{align} \label{B}
B =  H_1^{n-1}MH_2^{n-1} + H_1^{n-2}MH_2^{n-2}  + \ldots + H_1MH_2 + M.
\end{align}

The shared secret key $K$ is

\begin{align} \label{K}
K =  H_1^{n+m-1}MH_2^{n+m-1} + H_1^{n+m-2}MH_2^{n+m-2}  + \ldots + H_1MH_2 + M.
\end{align}

What makes our scheme compare favorably to, say, the scheme of \cite{Habeeb} is that in computing $A$ and $B$, both operations on matrices (addition and multiplication) are employed, which is good for security because neither multiplicative (e.g. the determinant) nor additive (e.g. the trace) functions of a matrix can be used to reduce the problem to a problem in $\Z_p$.

Also, the expressions for $A$ and $B$ cannot be factored into products of simpler expressions.
To compare, factoring would be possible if the action of (a single matrix) $H$ on $M$ was a one-sided multiplication. For example, if the action of $H$ on $M$ was given by $M \cdot H$, then the matrix $A$ would be equal to $MH^m + MH^{m-1}  + \ldots + MH + M = M(H^m + H^{m-1}  + \ldots + H + I)$. The adversary could multiply this on the right by $(H-I)$ and get $M(H^{m+1}-I)$.  Since the matrices $M$ and $H$ are public, the adversary could then recover $MH^{m+1}$ or even $H^{m+1}$ if $M$ is  invertible. Then, if  $H$ is  invertible, too, the adversary can recover  $H^{m}$, and this breaks the scheme.

However, with the action $M \to H_1MH_2$, no factorization of $A$ (or $B$) is possible, and there is no visible way to recover $K$ from $(M, H_1, H_2, A, B)$.

\subsection{Discrete log problem in $\Z_p$ as a special case}\label{log}

We note again that no power of any matrix is ever exposed, so all standard attacks on Diffie-Hellman-like protocols are not applicable in our situation. We also note that the analog of the discrete logarithm problem for our protocol is at least as hard as it is for the classical Diffie-Hellman protocol.

\begin{proposition}\label{reduction}
Let a prime $p$ be of the form $p=4n+3$. (In particular, safe primes have this property.) Suppose it is computationally feasible for the adversary to recover, from the $3 \times 3$ public matrices $A$, $M$, $H_1$, and  $H_2$ over $\Z_p$, the private exponent  $m$ in the protocol in Section \ref{Protocol}. Then it is computationally feasible for the adversary to recover the private exponent  $k$ from $g$ and $g^k$ in the classical Diffie-Hellman protocol.
\end{proposition}

\begin{proof}
In the notation we used for the  Diffie-Hellman protocol in the Introduction, suppose we are given $g$ and $g^m$, for some $g \in \Z_p$ and $m \in Z$. Make a $3 \times 3$ matrix $M$ with $a_{11}= g$; $a_{22}$, $a_{23}$,  and $a_{33}$ random. Let the matrix $H=H_1=H_2$ be a diagonal matrix, with  $h_{22}=1$; $h_{11}$ random, and $h_{33}=0$.

Then, for any $i\ge 0$, the matrix $H^i$ will have $h_{11}^i$, 1, and $0$ on the diagonal. The matrix $H^iMH^i$ will have $g^{2i}$ in the upper left corner, and zeros as other entries in the first row and in the first column. Therefore, the matrix $A$ (see the formula (\ref{A})), which is the sum of $H^iMH^i$, will have
$g^{2\sum i}$ in the upper left corner, and zeros as other entries in the first row and in the first column. The summation in the exponent on $g$ is from $i=0$ to $i=m-1$ and the whole exponent on $g$ is therefore equal to $(m-1)(m-2)$. Clearly, one knows $m$ if and only if one knows $(m-1)(m-2)$. Other entries of the matrix $A$ are as follows. The entry (2,2) is $m \cdot a_{22}$, i.e., it is random since $a_{22}$ is random. Similarly, the entry (2,3) is $m \cdot a_{23}$, i.e., is random. The entry (3,2) is 0. Finally, the entry (3,3) is $a_{33} \cdot  \sum h_{33}^j$, where $j$ in the exponent on $h_{33}$ runs from $j=0$ to $j=m-1$. Again, this is  random since $a_{33}$ is random.

Thus, the matrix $A$ has $g^{(m-1)(m-2)}$ in the upper left corner, and other entries are either zeros or random. Hence our strategy for recovering $k$ from $g$ and $g^k$ will be as follows. Put $g^k$ in the upper left corner of $A$, and make other entries either zeros or random, according to what we wrote in the previous paragraph.
Then recover $m$, if possible, using an algorithm that is assumed to exist by the Proposition hypothesis. We say ``if possible" because $k$ may not have the form $(m-1)(m-2)$, in which case the algorithm may fail.

Re-write $(m-1)(m-2)$ as $(m - \frac{3}{2})^2 - \frac{9}{4}$. Here $\frac{3}{2}$ means $3 \cdot 2^{-1}$. The inverse of 2 exists since $p \ne 2$. Similarly, $\frac{9}{4}$ is a particular fixed element of $\Z_p$; denote it by $r$. Thus, $k-r = (m - \frac{3}{2})^2$, i.e., $(k-r)$ should be a quadratic residue modulo  $p$ for our hypothetical algorithm to work. If $(k-r)$ is not a quadratic residue modulo  $p$, then, since $p$ is of the form $p=4n+3$, it is known that $(r-k)$ should be a quadratic residue. Then we replace $g^{k}$ in the matrix $A$ by $g^r (g^{k})^{-1} = g^{r-k}$ and run our algorithm on this updated matrix as well.

Thus, we will have two similar algorithms running in parallel on two different matrices $A$: one with $g^{k}$ in the upper left corner, the other one with $g^{r-k}$. One of these algorithms will recover $m$, and therefore $k$ (or $(r-k)$), and this completes the proof.

\end{proof}

\subsection{Why we need the condition $det(H_1H_2)=0$}\label{det} \hfill

\noindent The following attack was suggested by Dan Brown.

Let $d=det(H_1H_2)$. Then, since
$A = M + H_1 M H_2 +....+  H_1^{m-1} M H_2^{m-1}$, one gets
$$(H_1 A H_2 + M - A) \cdot M^{-1} = H_1^m M H_2^m \cdot  M^{-1}.$$

The matrix on the left is public, so its determinant can be computed by the attacker. The determinant of the matrix on the right is $d^m$. Therefore, the attacker who can solve the discrete log problem in $\Z_p$ can recover $m$ from $d$ and $d^m$. This does not work with $d=0$ though, which is why we impose this condition. We actually impose a stronger condition $det(H_1)=det(H_2)=0$ to avoid attacks similar to linear algebra attacks \cite{Stickel} on Stickel's key exchange scheme \cite{Stickel_self}.

\subsection{Indistinguishability from random}\label{random}

We have run some tests to see if the matrix $K$ (the shared secret key) is indistinguishable from random. Figure 1 shows a histogram of values of the (1,1) entry of $K$, for a 200-bit $p$. Here values from 0 to $p$ are split into 10 groups (``bins") of size $\left \lfloor \frac{p}{10} \right \rfloor$ and the number of values in each bin, out of 100,000 trials, is recorded. The histogram shows essentially uniform distribution of values between the bins. Histograms for values of other entries of the matrix $K$ look the same, so $K$ passes at least this simple randomness test.

\begin{figure}[h!]
\vskip -0.5cm
\includegraphics[width=3in, height=3in]{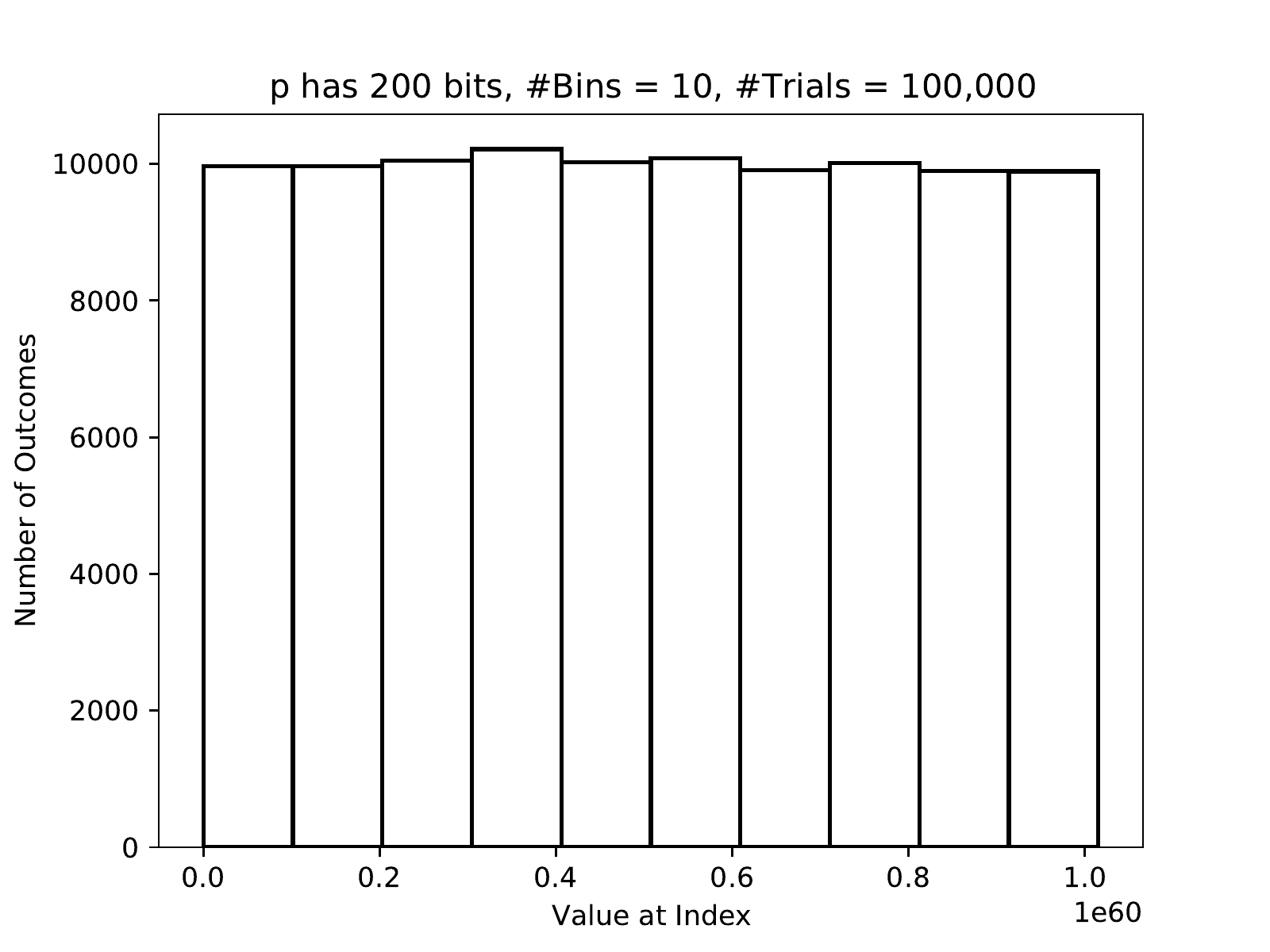}
\vskip -0.5cm \caption{Values distribution for the (1,1) entry of $K$} \label{bins}
\end{figure}

Another test we have run was computing the mean of the entries of $K$ in each single row and each single column, as well as the mean of {\it all} entries of the matrix $K$. Figure 2 below shows a histogram of the means of the entries of the first column, out of 100,000 trials. Figure 3 shows a histogram of the means of all entries of the matrix $K$. Again, values from 0 to $p$ are split into 10 groups (``bins") of size $\left \lfloor \frac{p}{10} \right \rfloor$. By the central limit theorem, if several random variables are independent and identically distributed, then their mean is approximately normally distributed. Thus, our Figure 2 and Figure 3 are, though  indirect, still an evidence of different entries of $K$ being independent.

\begin{figure}[h!]
\includegraphics[width=3in, height=3in]{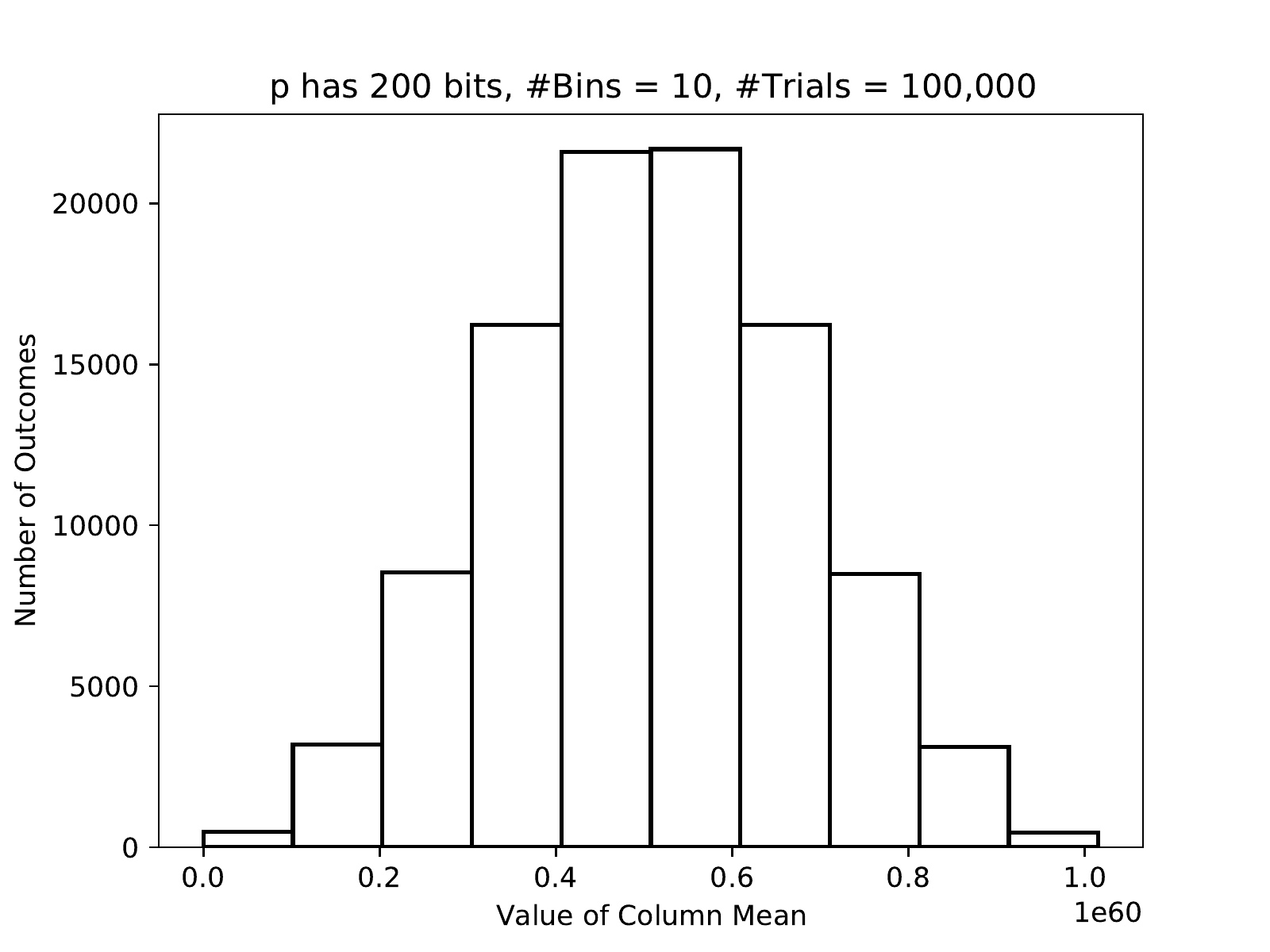}
\vskip -0.5cm \caption{Mean distribution of the first column entries of $K$} \label{bins} \hskip 0.1in \includegraphics[width=3in, height=3in]{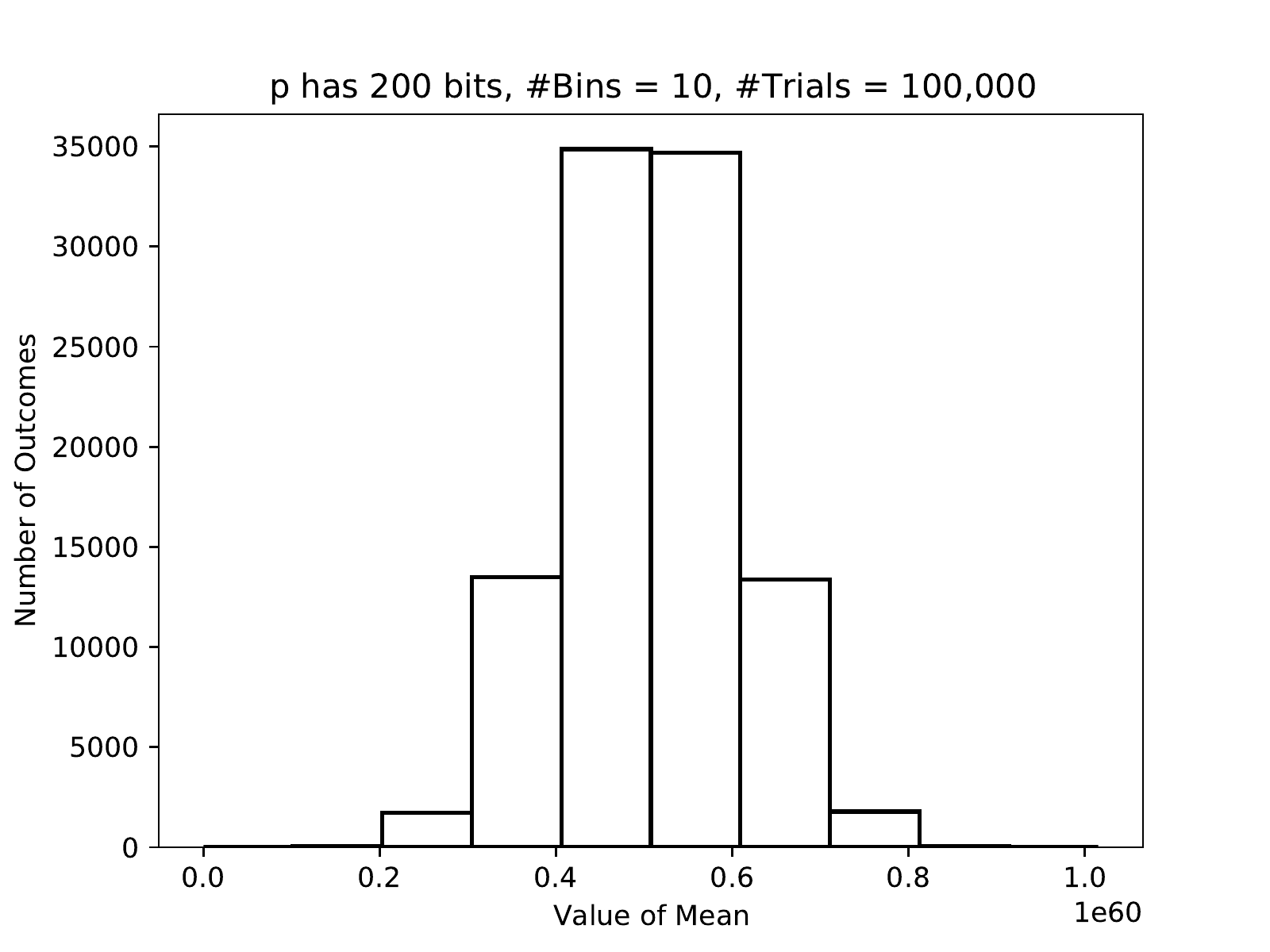} \vskip -0.5cm \caption{Mean distribution of all entries of $K$}
\end{figure}

Another evidence of the matrix $K$ being indistinguishable from  random is independent distribution of values of different entries of $K$. This is illustrated by Figure 4. It shows that joint distribution of values of a pair of different entries of $K$ is very close to uniform, i.e., for any two possible values $(x, y)$ of the entries in such a pair, the probability to occur is  $\frac{1}{p^2}= \frac{1}{p} \cdot \frac{1}{p}$, evidencing independence of $x$ and $y$. The ticks on the x-axis of the histogram in Figure 4 split $p$ possible values of the first entry in a pair into 10 bins corresponding to sets of values with an increment of $h = \left \lfloor \frac{p}{10} \right \rfloor$. Each of these 10 bins is split again into 10 bins of equal size to accommodate possible values of the second entry in a pair.

\begin{figure}[h!]
\vskip -0.2cm
\includegraphics[width=3in, height=3in]{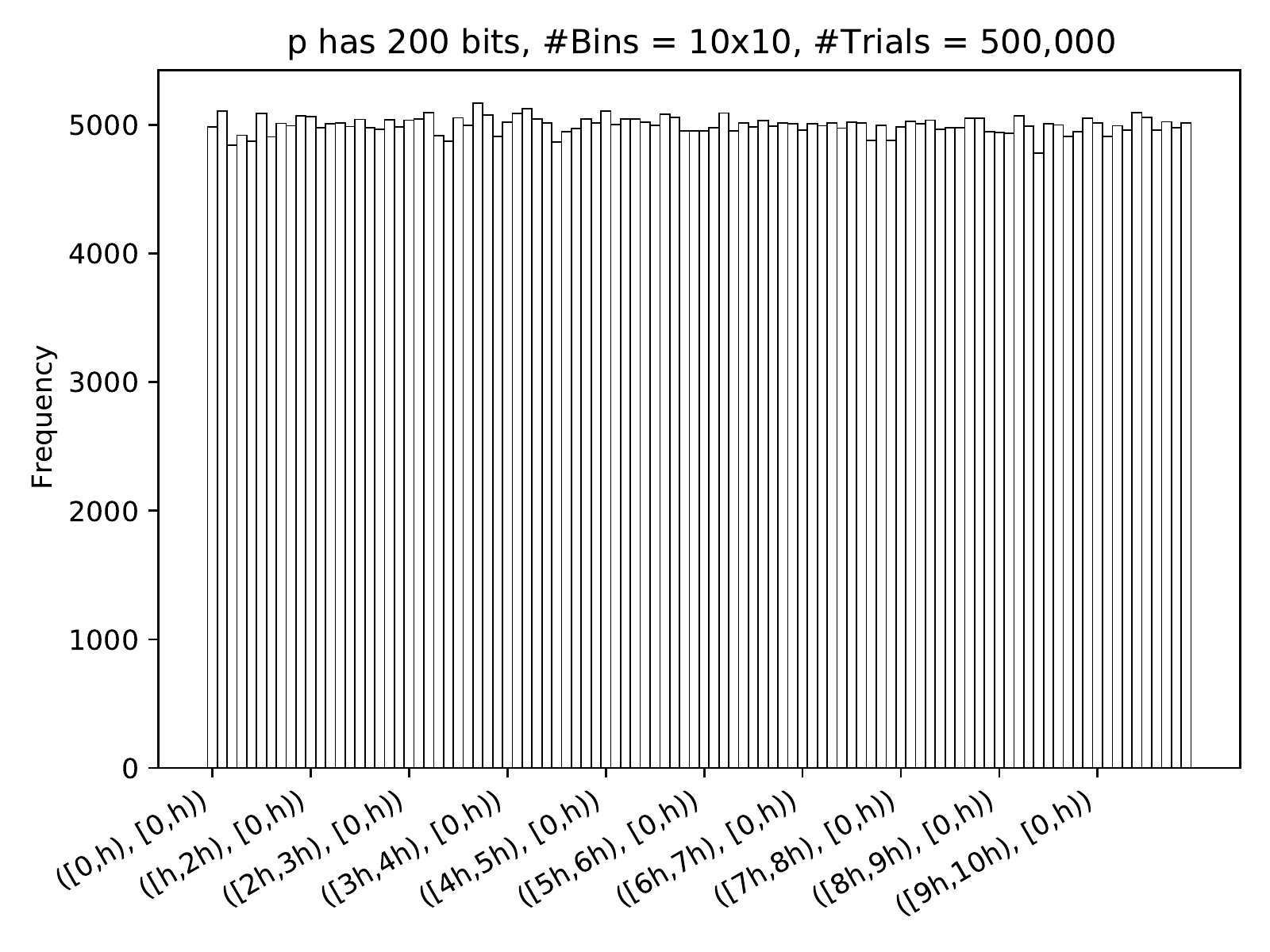}
\vskip -0.5cm \caption{Distribution of values of a pair of entries of $K$} \label{bins}
\end{figure}

\section{Implementation and performance}
\label{Implementation}

The scheme of this paper was implemented using Python. The code is available online, along with a challenge, see \cite{Python}.

We note that the formulas (\ref{A}) and (\ref{B}) may create an impression that to compute the matrices $A$
and $B$, Alice and Bob have to compute all powers of the matrices $H_1$ and $H_2$, from 1 to  $(m-1)$ or $(n-1)$. However, this is not the case; a large power of $(M, (H_1, H_2))$ can be computed with the usual square-and-multiply method, and this will produce $A$ or $B$ as the corresponding first component.

The number of multiplications in $\Z_p$ needed to compute $H_i^n$ is, of course, larger than that to compute $g^n$ for $g \in \Z_p$. To compute the square of a  $3 \times 3$ matrix, one needs 24 multiplications of elements of $\Z_p$, so one can expect our protocol to run at least 24 times slower than the classical Diffie-Hellman protocol. This is still fast enough to be practical; with a 2000-bit $p$ the run time of the protocol is about 2 sec on a very basic computer. Besides, given that at this time there are no visible approaches whatsoever (other  than brute force) to attack the protocol in this paper, the basic parameter $p$ can probably be taken smaller than what is recommended for standard Diffie-Hellman protocols,  and this will reduce the run time.

A particular 2000-bit safe prime $p$ that we used in our computer simulations was
10045850546888 50036334185776562243339025531704844369832736073099638458477395071158608659647532399390\\
27972338834707903941940188314348678981808910413754306718965087266944429878241410578991\\
73376250244281758576559881643143110828207143325627334593997352683778809319929255772120\\
45905540615043591215742223683070489198090104899809610177067292220347910171309250704268\\
93349814057145812995340991548906078333104951440614482037356443864699967124299012034397\\
81034231264233355059817445403969916571063605224058329470399818911447991765712527069708\\
6234200442489544474659560583354052797579309573507121265302226528942789519.

\section{Conclusions}

$\bullet$ We have offered  a key exchange protocol, resembling the classical Diffie-Hellman protocol, based on a semidirect product of two cyclic semigroups of matrices over $\Z_p$. One of the semigroups is additive, the other one multiplicative.

$\bullet$  In our protocol, no power of any matrix or any element of $\Z_p$ is ever exposed, so all standard attacks on Diffie-Hellman-like protocols are not applicable.

$\bullet$ Security assumption, analogous to the {\it computational Diffie-Hellman assumption}, is computational infeasibility of recovering the value of a  polynomial (of unknown degree) of a special form on three given matrices, from the values of two other polynomials (also of unknown degrees) on the same three matrices. A weaker security assumption is analogous to the discrete log assumption in  $\Z_p$ and is provably at least as hard.

\vskip 0.2cm

\noindent {\bf Acknowledgement.} We are grateful to Neal Koblitz who pointed out a ``diagonalization attack" on an earlier version of our protocol where the matrices $H_1$ and $H_2$ were equal. This attack reduces the problem of recovering a private exponent in our protocol (with $H_1=H_2=H$) to the standard discrete log problem in a field extension of $\F_p$. We are also grateful to Dan Brown for suggesting the attack in Section \ref{det}.

\end{document}